%% file: bundle.tex
\RequirePackage{snapshot}
\documentclass[]{llncs}
\usepackage{graphicx}
\usepackage{hyperref}
\hypersetup{
  unicode=true,
  pdfkeywords={session types; applied type system; dependent types; linear types; multirole logic; concurrency; deadlock-free},
  pdftitle={Multiparty Dependent Session Types \break\small (Extended Abstract)},
  pdfauthor={},
  colorlinks=true,
  linkcolor=red,
  citecolor=green,
  urlcolor=cyan,
  breaklinks=true
}

\usepackage{mathrsfs}
\usepackage[intlimits]{amsmath}
\usepackage{amsfonts,amssymb,mathtools}
\DeclareSymbolFontAlphabet{\mathbb}{AMSb}

\usepackage{varwidth}
\usepackage{wrapfig}

\usepackage{todonotes}

\usepackage{longtable,booktabs}

\usepackage{colonequals} 
\usepackage[a]{esvect}   
\usepackage{prftree}
\usepackage{adjustbox}   
\usepackage{pgfmath}
\usepackage{xfrac}
\usepackage{stmaryrd}    
\usepackage{inconsolata} 
\usepackage{array}       
\usepackage{rotating}    
\usepackage{pdflscape}   
\usepackage{soul}        

\spnewtheorem*{subcase}{Subcase}{\itshape}{\rmfamily}

\usepackage{tikz}
\usepackage{tikz-cd}
\usetikzlibrary{arrows,shapes,petri,decorations.pathmorphing}
\tikzcdset{arrows={line width=1.414*rule_thickness}}

\usepackage{listings}
\usepackage[capitalise,noabbrev,nameinlink]{cleveref} 
\crefname{sec}{Section}{Sections}
\crefname{prop}{Proposition}{Propositions}
\crefname{lemma}{Lemma}{Lemmas}
\crefname{def}{Definition}{Definitions}
\crefname{thm}{Theorem}{Theorems}
\crefname{app}{Appendix}{Appendices}
\crefname{eg}{Example}{Examples}
\crefname{lst}{Listing}{Listings}
\crefname{remark}{Remark}{Remarks}

\makeatletter
\newcommand{\leqnomode}{\tagsleft@true}
\newcommand{\reqnomode}{\tagsleft@false}
\makeatother

\IfFileExists{upquote.sty}{\usepackage{upquote}}{}

\providecommand{\tightlist}{\setlength{\itemsep}{0pt}\setlength{\parskip}{0pt}}

\begin{document}
\author{Hanwen Wu\inst{1} \and Hongwei Xi\inst{1}}
\institute{Boston University, Boston MA 02215, USA \\
           \email{\{hwwu,hwxi\}@bu.edu}}

\newcommand{\rulename}[1]{\text{{\upshape\bf #1}}} 
\newcommand{\ptitle}[1]{{\bf #1}\kern.5em}
\def\ATS{\mathcal{A}\mathcal{T}\kern-.1em\mathcal{S}}
\def\zmq{\O{}MQ}
\def\defeq{\coloncolonequals}  

\def\Lzero{\lambda_0}
\def\Lpi{\lambda_0^{\pi}}
\def\Ldep{\lambda_{\forall,\exists}}
\def\Ldeppi{\lambda_{\forall,\exists}^{\pi}}

\newcommand{\sta}[1]{{\color{blue!60!black}{#1}}}

\def\sortt(#1){\sta{{\color{black}{\text{\rm\it #1}}}}} 
\def\sortm(#1){\sta{{\color{black}{#1}}}}               

\def\sorttc(#1){\text{\rm\it #1}}

\def\stset{\sortt(set)}
\def\stint{\sortt(int)}
\def\stbool{\sortt(bool)}
\def\sttype{\sortt(type)}
\def\stvtype{\sortt(vtype)}
\def\ststype{\sortt(stype)}
\def\strole{\sortt(role)}
\def\stsigma{\sortm(\sigma)}
\def\stnat{\sortt(nat)}

\def\stsetc{\text{\rm\it set}}
\def\stintc{\text{\rm\it int}}
\def\stboolc{\text{\rm\it bool}}
\def\sttypec{\text{\rm\it type}}
\def\stvtypec{\text{\rm\it vtype}}
\def\ststypec{\text{\rm\it stype}}
\def\strolec{\text{\rm\it role}}
\def\stnatc{\text{\rm\it nat}}

\newcommand{\stermt}[1]{\sta{\text{\rm\it #1}}\kern0.1em} 
\newcommand{\stermm}[1]{\sta{#1}}               
\newcommand{\stermtc}[1]{\text{\rm\it #1}}      
\newcommand{\stermmc}[1]{#1}      				

\def\subtype{\mathbin{\stermm{\leq_{ty}}}}
\def\sassert{\mathbin{\stermm{\wedge}}}
\def\sguard{\mathbin{\stermm{\supset}}}

\def\sr{\stermm{r}}
\def\snegr{\stermm{\neg r}}

\def\strue{\stermm{\top}}
\def\sfalse{\stermm{\bot}}

\def\rs{\stermtc{rs}}   
\def\rsfull{\rs_{\fullset}}
\def\scx{\stermt{scx}}
\def\scc{\stermt{scc}}
\def\scf{\stermt{scf}}


\def\typet(#1){\stermt{#1}}  
\def\typem(#1){\stermm{#1}}
\def\typetc(#1){\stermtc{#1}}  
\def\typemc(#1){\stermmc{#1}}

\newcommand{\chan}[2]{\typem({\typet(chan)({#1},\stypem({#2}))})}
\newcommand{\chanc}[2]{\typemc({\typetc(chan)({#1},\stypemc({#2}))})}

\def\sunit{\typet(unit)} 
\def\sunitc{\typetc(unit)}
\def\sint{\typet(int)}
\def\sintc{\typetc(int)}
\def\sintn(#1){\sta{\sint(\stermm{#1})}}
\def\sstring{\typet(string)}
\def\sstringc{\typetc(string)}
\def\sbool{\typet(bool)}
\def\sboolb(#1){\sta{\sbool(\stermm{#1})}}
\def\stau{\typem(\tau)}
\def\shattau{\typem(\hat\tau)}

\def\stypet(#1){{\color{cyan}{\texttt{#1}}}}  
\def\stypem(#1){{\color{cyan}{#1}}} 		  
\def\stypemc(#1){{#1}} 		  

\def\stype{\stypet} 
\def\spi{\stypem(\pi)}

\def\msg{\stypet(msg)}
\def\seqs{\mathbin{\stypem(\coloncolon)}}
\def\quan{\stypet(quan)}
\def\fix{\stypet(fix)}
\def\nil{\stypet(end)}
\def\branch{\stypet(branch)}
\def\ite{\stypet(ite)}
\def\init{\stypet(init)}

\newcommand{\dyn}[1]{{\color{red!50!black}{#1}}}

\def\dkw(#1){\text{\rm\bf #1}}

\def\dlam{\dkw(lam)~}    
\def\dfix{\dkw(fix)~}    
\def\dapp{\dkw(app)}     
\def\dfst{\dkw(fst)}     
\def\dsnd{\dkw(snd)}     
\def\din{~\dkw(in)~}     
\def\dif{\dkw(if)~}      
\def\dthen{~\dkw(then)~} 
\def\delse{~\dkw(else)~} 

\newcommand{\dtermt}[1]{\dyn{\text{\rm #1}}} 
\newcommand{\dtermm}[1]{\dyn{#1}}            

\def\dtrue{\dtermt{true}}
\def\dfalse{\dtermt{false}}
\def\dpair(#1,#2){\dtermm{\langle{#1},{#2}\rangle}}           
\def\dunit{\dtermm{\langle\rangle}}                           
\newcommand{\dite}[3]{\dtermm{\dif{#1}\dthen{#2}\delse{#3}}}  
\newcommand{\dlet}[3]{\dtermm{\dkw(let)~{#1}={#2}\din{#3}}}   

\def\dguardi{\dtermm{\supset^+}}  
\def\dguarde{\dtermm{\supset^-}}
\def\dassert{\dtermm{{\wedge}}}
\def\dforalli{\dtermm{\forall^+}}
\def\dforalle{\dtermm{\forall^-}}
\def\dexists{\dtermm{{\exists}}}

\def\dcx{\dtermt{dcx}}
\def\dcc{\dtermt{dcc}}
\def\dcf{\dtermt{dcf}}
\def\dcr{\dtermt{dcr}}
\def\dxf{\dtermm{\text{\rm\it xf}}}
\def\dtid{\dtermt{tid}}
\def\dpool{\dtermm{\Pi}}


\def\dapi(#1){\textrm{\color{red}{#1}}}

\def\dfork{\dapi(fork)}
\def\dsend{\dapi(send)}
\def\drecv{\dapi(recv)}
\def\dbsend{\dapi(bsend)}
\def\dbrecv{\dapi(brecv)}
\def\dskip{\dapi(skip)}
\def\dclose{\dapi(close)}
\def\dwait{\dapi(wait)}
\def\dcut{\dapi(cut)}
\def\dappend{\dapi(append)}
\def\dunify{\dapi(unify)}
\def\dexify{\dapi(exify)}
\def\delim{\dapi(elim)}
\def\dsplit{\dapi(split)}
\def\drecurse{\dapi(recurse)}
\def\doffer{\dapi(offer)}
\def\dchoose{\dapi(choose)}
\def\daccept{\dapi(accept)}
\def\drequest{\dapi(request)}


\def\res{{\color{black}{\mathcal{R}}}}      
\def\sig{{\color{black}{\mathcal{S}}}}      
\def\ectx{{\color{black}{E}}}               
\def\cctx{{\color{black}{C}}}               

\def\map(#1,#2){[{#1}\mapsto{#2}]}         
\newcommand{\subst}[2]{[\sfrac{#2}{#1}]}   


\def\reduce{\mathrel{\longrightarrow}}        
\def\reduceall{\mathrel{\longrightarrow^*}}
\def\vreduce{\stackrel{v}{\longrightarrow}}   
\def\creduce{\stackrel{P}{\reduce}}           
\def\creduceall{\mathrel{\overset{P}{\reduce}\!\!^*}}
\def\dfreduce{\mathrel{\leadsto}}             
\def\dfreducevia(#1){\stackrel{\dtermm{{#1}}}{\dfreduce}} 
\def\dfnormal{\not\dfreduce}
\newcommand{\redex}[1]{\underline{#1}}


\def\channels{\textsf{\textup{channels}}}
\def\endpoints{\textsf{\textup{endpoints}}}
\def\roles{\textsf{\textup{roles}}}

\def\fv{\textit{fv}\kern0.1em}        
\def\fn{\textit{fn}\kern0.1em}        
\def\dom{\textit{dom}\kern0.1em}      

\def\D{\mathscr{D}}         
\def\height{\textit{ht}}    

\def\M{\mathcal{M}}        


\def\ploop{\textsf{\textup{loop}}}
\def\pblocked{\textsf{\textup{blocked}}}
\def\pregular{\textsf{\textup{regular}}}
\def\pdfreducible{\textsf{\textup{df-reducible}}}
\def\pnormal{\textsf{\textup{df-normal}}}
\def\pdeadlocked{\textsf{\textup{deadlocked}}}
\def\pconsistent{\textsf{\textup{consistent}}}
\def\pmatch{\textsf{\textup{match}}}
\def\prelaxed{\textsf{\textup{relaxed}}}


\def\tensor{\varotimes}
\def\parr{\bindnasrepma}
\def\with{\binampersand}

\def\f{{\mathscr{F}}}
\def\uf{{\mathscr{U}}}
\def\conju{\uf}
\def\forallu{\uf^{\lambda}}
\def\mconju{\uf^\times}
\def\aconju{\uf^+}
\def\expou{\uf^\ast}
\def\impliesu{\uf^f}

\def\fullset{\overline{\varnothing}}
\newcommand{\inv}[1]{{#1}^{\raisebox{.2ex}{$\scriptscriptstyle-\!1$}}}

\newcommand{\iform}[2]{[{#1}]_{{#2}}}
\newcommand{\iformc}[2]{\cp{\iform{#1}{#2}}}


\def\cli{\stermm{C}}
\def\srv{\stermm{S}}



\def\myspacinggather{\setlength{\jot}{1em}}
\title{Multiparty Dependent Session Types \break\small (Extended Abstract)}
%


%
\maketitle              
%

\begin{abstract}
Programs are more distributed and concurrent today than ever before, and
structural communications are at the core. Constructing and debugging
such programs are hard due to the lack of formal
specification/verification of concurrency. This work formalizes the
first \emph{multiparty dependent session types} as an expressive and
practical type discipline for enforcing communication protocols. The
type system is formulated in the setting of multi-threaded
\(\lambda\)-calculus with inspirations from multirole logic, a
generalization of classical logic we discovered earlier. We prove its
soundness by a novel technique called deadlock-freeness reducibility.
The soundness of the type system implies communication fidelity and
absence of deadlock.

\end{abstract}


%

\hypertarget{introduction}{%
\section{Introduction}\label{introduction}}

\emph{Session type} \cite{Honda:1993eh,Takeuchi:1994bv,Honda:1998fm} is
a typed formalism for concurrency. A \emph{session} is an abstraction of
structured communication among two or more logical \emph{parties}
connected by a communication \emph{channel}. Session types denote the
structures of communications, or \emph{protocols}, and are assigned to
communication channels. In a typical session type system, subjection
reduction ensures \emph{session fidelity} and progress property ensures
\emph{absence of deadlock}. As a result, well-session-typed programs
cannot make communication errors.

In this work, we present our research results on a practical multiparty
dependent session type system. It is a system that can describe more
than two participants (multiparty), that supports quantification and
polymorphism in the session type (dependent), and that it is formulated
in the settings of multi-threaded \(\lambda\)-calculus (practical).
Other features include higher-order sessions (sending channels over
channels), forwarding (connecting channels with channels), recursive
sessions (as an extension), etc. A well-session-typed program strictly
follows the session protocol (subject reduction) and is absent of
deadlock (progress). To the best knowledge of the authors, this is the
first formulation of a multiparty dependent session type system.

\hypertarget{simple-examples}{%
\subsection{Simple Examples}\label{simple-examples}}

We first fix some terminologies. A \emph{session} has several
\emph{parties} connected via a \emph{channel}. \emph{Session types}
encodes communication structures \emph{globally}. A channel has many
endpoints. Each party is (typically) implemented as a thread, and each
thread function is given an \emph{endpoint}. Endpoints are assigned an
\emph{endpoint type} representing the \emph{local protocl} from the
perspective of this particular party.

\begin{example}
\label[eg]{eg:helloworld} A simple ``Hello World'' protocol
\stypet(hello) between server \(\srv\) and client \(\cli\), can be
described using session types as follows. \[
\stypet(hello)
\defeq
\stypem({\msg(C,\sstringc) \seqs \msg(S,\sstringc) \seqs \nil(C)})
\] This protocol specifies the communications globally, where \(\cli\)
first sends a message of type \(\sstring\), followed by \(\srv\) sending
also a \(\sstring\), followed by \(\cli\) terminating the session while
\(\srv\) waits for the termination. Locally at each party, \(\cli\)
holds an endpoint of type \(\chan{C}{\stypet(hello)}\), while \(\srv\)
holds an endpont of type \(\chan{S}{\stypet(hello)}\) where the linear
type constructor \(\stermt{chan}\) combines \(\stypet(hello)\) with a
role \(\cli\) or \(\srv\) to form local endpoint types. A program for
\(\cli\) can be written as \(\dtermt{cli}\), while the server is
\(\dtermt{srv}\). \begin{align*}
\dtermt{cli} & \defeq\dtermm{
    \dlam c.\dlet{c}{\dbsend(c,\dtermt{'hello'})}{}
    \dlet{\dpair(c,\dtermt{rpl})}{\dbrecv(c)}{}
    \dclose(c)
}\\
\dtermt{srv} & \defeq\dtermm{
    \dlam s.\dlet{\dpair(s,\dtermt{req})}{\dbrecv(s)}{}
    \dlet{s}{\dbsend(s,\dtermt{'world'})}{}
    \dwait(s)
}\\
\dtermt{pool}&\defeq\dtermm{
    \dlet{s}{\dfork(\dtermt{cli})}{\dapp(\dtermt{srv}, s)}    
}
\end{align*} \(\dbsend\)/\(\dbrecv\) etc,. are Session APIs provided by
our type system used to realize the dynamic semantics of session types
by interpreting them locally at each party. The type system will
guarantee that the correct API is invoked at the correct stage of
protocol in the correct order and that all endpoints are invoking
\emph{dual/compatible} APIs in order to make progress. Finally in
\(\dtermt{pool}\), \(\dfork\) spawns a new thread with thread function
\(\dtermt{cli}\), and connects to that thread with a session typed
channel while returning the other endpoint \(\dtermm{s}\) to the caller.
\end{example}

\begin{example}
\label[eg]{eg:array} With quantification in the session types, one can
safely send an array by firstly sending a length \(n\) followed by \(n\)
repeated messages for \(n\) elements of the array. \[
\begin{array}{rcl}
\stypem({\stypet(array)(\tau{:}\sttypec)})  &\defeq& \stypem({\quan(C, \lambda n{:}\stintc.\msg(C, \sintc(n))\seqs\stypet(repeat)(\tau, n))})
\end{array}
\] In the above definition, \(\sintn(n)\) is a singleton dependent type
for an integer of value \(n\), \(\quan\) is a session type constructor
that represents a quantifier, where \(\stermm{\lambda n{:}\stintc}\) is
the actual binder. The quantifier in this case will be interpreted as
universal by \(\cli\), and existential by all others. For instance, the
endpoint at \(\cli\) will have linear type
\(\chan{C}{\stypet(array)(\tau)}\), and after invoking API \(\dunify\)
on this endpoint, the type becomes
\(\stermm{\forall n{:}\stintc.\chan{C}{\msg(C,\sintc(n))\seqs\stypet(repeat)(\tau,n)}}\).
The bound variable \(\stermm{n}\) ensures that the length of the array
equals the number of repeated messages that follows.
\end{example}

\hypertarget{contributions}{%
\subsection{Contributions}\label{contributions}}

The main contribution lies in the formalization of multiparty dependent
session types and its deadlock-freeness proof via a novel technique
named deadlock-freeness reducibility. We summarize the contributions as
follows.
\begin{itemize}
\tightlist
\item
  Formalized the first multiparty dependent session type system
  (\(\Ldeppi\)) and proved its soundness.
\item
  Formalized deadlock-freeness reducibility. It is a pool reduction
  invariant, even in the presence of higher-order sessions and various
  forms of forwarding. The progress property directly dependents on the
  df-reducibility of thread pools.
\item
  Discovered and formulated classical multirole logics (MRL) and linear
  multirole logic (LMRL) as generalizations of classical logic (LK) and
  classical linear logic (CLL). We proved the admissibility of a cut
  rule that combines more than two sequents in both MRL and LMRL, thus
  generalizing Gentzen's celebrated results of cut-elimination.
\item
  We report that we have work-in-progress implementations. Due to space
  limits, we omit this part and please see
  \url{http://multirolelogic.org}.
\end{itemize}
\hypertarget{overview}{%
\subsection{Overview}\label{overview}}

This work has three parts, and we focus mostly on the third part below.
\emph{First}, our technical foundation is \(\ATS\) \cite{Xi:2003kl}. We
very briefly mention the approaches to formulating types, reasoning
about resources, and adding pre-defined functions in \(\ATS\).
\emph{Second}, we mention the intuitions of MRL/LMRL and present the
generalized cut rule combining more than two sequents. We discuss how
LMRL deeply influenced the design of \(\Ldeppi\). \emph{Third}, we
formulate \(\Ldeppi\) in \(\ATS\) and show the deadlock-freeness
reducibility proof.

\hypertarget{applied-type-system}{%
\section{Applied Type System}\label{applied-type-system}}

\(\Ldeppi\) is built on Applied Type System (\(\ATS\)), a multi-threaded
\(\lambda\)-calculus with advanced types. \(\ATS\)
\cite{Xi:2003kl,Xi:2016vd} is the successor of Dependent ML
\cite{Xi:2007te,Xi:1998wa,Xi:1999bh}. As a general framework for
formalizing type systems, \(\ATS\) supports dependent types of
DML-style, linear types, theorem proving, general recursion, among other
features. Due to space limits, we will only cover selected features of
\(\ATS\) to prepare for the development of \(\Ldeppi\). Please refer to
\cite{Xi:2003kl,Xi:2016vd} for a full treatment.

\hypertarget{preview}{%
\subsection{Preview}\label{preview}}

The key salient feature of \(\ATS\) lies in the complete separation
between \emph{statics}, where types are formed and reasoned about, and
\emph{dynamics}, where programs are constructed and evaluated.

Types in the dynamics are terms of the statics, where statics can be
regarded as a simply-typed \(\lambda\)-calculus whose ``types'' are
called \emph{sorts}. Types can depend on static terms (hence
\emph{dependent} types). In \(\Ldeppi\), we will be formulating session
types \(\spi\) in the statics, so that it can be used as an index for
the endpoint type, i.e., \(\chan{\rs}{\pi}\). After having terms of
endpoint types, we will provide programmers with pre-defined functions
called \emph{Session APIs} to manipulate them, like
\(\dbsend\)/\(\dbrecv\). They are formulated in the dynamics as constant
functions \(\dcf\) in \(\ATS\). Their pre-defined types, called dc-types
(dynamic constant types), will be carefully designed and recorded in a
pre-defined context called signature \(\sig\) to perform type-checking.
Their reductions in a thread pool are also formulated and reasoned
about.

Endpoints are resources, meaning they cannot be randomly copied or
discarded. Resource ownership has to be tracked by the type system to
prevent bugs like memory leaks, use-after-free, etc. The support for
linear types in \(\ATS\) provides a mechanism for reasoning about
resources. We will define endpoints as dynamic resources \(\dcr\),
define the function \(\rho(\cdot)\) (\cref{fig:mtlcallrho}) to collect
resources, define its consistency conditions to prevent ill-formed
channels (e.g., missing endpoints, duplicated endpoints), and prove that
resource consistency is an invariant during reduction. Linear types also
provide a way to track the progress of protocols. For instance, given an
endpoint of \(\chan{C}{\msg(C,\sintc)\seqs\pi}\), invoking \(\dbsend\)
at party \(\cli\) will send a message, \emph{consume this linear
endpoint}, and return the endpoint \emph{with a new type}
\(\chan{C}{\pi}\). And only this endpoint with the continuation of the
protocol is available for use in the typing context.

\hypertarget{ats}{%
\subsection{\texorpdfstring{\(\ATS\)}{\textbackslash{}ATS}}\label{ats}}

The development of \(\ATS\) is fairly standard. It has an ML-like syntax
(\cref{fig:mtlcallsyntax}), non-linear/linear split context typings
(\cref{fig:mtlcalltyping}), and call-by-value, left-to-right, reduction
semantics (\cref{fig:mtlcallred}). Since \(\ATS\) is not the
contribution of this work, we will only illustrate some concepts using
examples.

\input{mtlcallsyntax}
\input{mtlcalltyping}
\input{mtlcallrho}
\input{mtlcallred}

\begin{example}
Consider the dynamic term (in red) and its type (in blue) below. \[
\dtermm{\dlam x.\dlam y.x/y} 
: \stermm{\forall m{:}\stintc.\forall n{:}\stintc.(n\neq 0)\sguard(\stermt{int}(m),\stermt{int}(n))\rightarrow\stermt{int}(m/n)}
\]
\end{example}

The term represents a function that does integer division (with the
result rounded to the nearest integer.) Given a static integer
\(\stermm{i}\), \(\stermm{\stermt{int}(i)}\) is a singleton type
representing a dynamic integer whose value equals to \(\stermm{i}\),
where \(\stermt{int}\) is a type constructor of sort
\(\sortm({\stint\Rightarrow\sttype})\). We use \(P\) for propositions,
i.e.~static terms of sort \(\sortt(bool)\). Given \(\stermm{P}\),
\(\typem({P \sguard \tau})\) is a guarded type. Intuitively, if a value
\(\dtermm{v}\) is assigned a guarded type \(\typem({P\sguard\tau})\),
then \(\dtermm{v}\) can be used only if the guard \(\stermm{P}\) is
satisfied. The whole type is a universally quantified, guarded, function
type that reads, for all (quantified) static integers \(\stermm{m}\) and
\(\stermm{n}\) where \(\stermm{n\neq 0}\) (guarded), we form a function
type where given two integers whose values are \(\stermm{m}\) and
\(\stermm{n}\), returns a integer whose value is \(\stermm{m/n}\) where
\(/\) should be interpreted as a static integer division with the result
rounded to the nearest integer. Importantly, division-by-zero will be an
\emph{type error}.

Thread pool \(\dpool\) is a collection of mappings \(\dtermm{t{:}e}\)
from thread id to closed expressions. Typing judgement is of the form
\(\Sigma;\stermm{\vv P};\Gamma;\Delta\vdash\dtermm{e}:\typem(\hat\tau)\)
where \(\Sigma\) is the sorting context, \(\stermm{\vv P}\) is a set of
propositions representing constraints, \(\Gamma\) is the non-linear
typing context, and \(\Delta\) is the linear typing context. We may
simply write \(\vdash\dtermm{e}:\typem(\hat\tau)\) if all contexts are
empty. Type equality is defined in terms of subtyping relations.
Type-checking is reduced into constraint-solving in some
constraint-domain.

Lastly, we mention that \(\ATS\) is sound. Please refer to
\cite{Xi:2003kl} for the proof of subject reduction and the progress
property.

\hypertarget{multirole-logic}{%
\section{Multirole Logic}\label{multirole-logic}}

Along the lines of
\cite{Abramsky:1993io,Abramsky:1994ez,Bellin:1994ua,Wadler:2012ua,Caires:2010gi}
that interpret cut reductions (i.e., cut-elimination steps) as
communications between \emph{two} parties in some session-typed process
calculi, we seek to find a logic that admits a cut rule combining
\emph{more than two} sequents as a foundation for \emph{multiparty}
session types. The notion of \emph{duality}, which is an inexplicit side
condition for the traditional cut rule, has to be generalized to account
for the \emph{coherence}/\emph{compatibility} of multiple sequents in a
cut, too. This led us to the discovery of multirole logic. Giving a full
account of MRL in this work is infeasible. We only present our insights
relevant to \(\Ldeppi\). \emph{In short, all the propositions in the
guarded types of Session APIs are directly influenced by the side
conditions of inference rules of MRL/LMRL.} Please refer to
\cite{Xi:2017wv} for details.

\input{mrlintuition}

The intuition is best summarized in \cref{fig:mrlintuition} with
selected rules from two-sided, one-sided, and ``many-sided'' sequent
calculus for classical logic. The rules for the two-sided sequent
calculus is well-known to be symmetric. The \rulename{$\neg$L} and
\rulename{$\neg$R} rules provide a way to move a formula from one side
to the other while remembering how many times a formula has been moved.
Due to involutive negation, we have the equivalent one-sided
presentation, where formulas are identified up to de Morgan duality. One
possible explanation for this duality is to think of the availability of
two roles 0 and 1 s.t. the left side of \(\vdash\) plays role 1 while
the right side does role 0. Negation is still about changing
roles/sides.

With this explanation, we can write the \rulename{id} rule in the
following ways in \cref{fig:mrlintuition}. In \rulename{id-two-sided}
and \rulename{id-one-sided}, the subscript 0 and 1 denotes sides, and
\(\vdash\) both separate sides and denotes derivability. Note how
negation changes the side of \(A\) from 1 to 0 while remembering it has
been moved once. In \rulename{id-two-sided-on-one-side}, we still have
two sides, denoted by the subscripts 0 and 1, except that we write them
both on the right of \(\vdash\). In this case, \(\vdash\) no longer
separate sides. It merely is a meta-symbol denoting the derivability of
formulas on its right. Using the style of
\rulename{id-two-sided-on-one-side}, it seems entirely natural for us to
introduce more roles into classical logic like in this three-sided
sequent \(\vdash\iform{A}{0},\iform{A}{1},\iform{A}{2}\). This leads us
to multirole logic.

\cref{fig:mrl} presents the syntax and inference rules of (classical)
MRL. A formula, combined with a role set \(R\) (represented as a set of
integers) within which the formula should be interpreted, is an
\(i\)-formula. Note that MRL is parameterized by some underlying
\emph{full} set of roles \(\fullset\), potentially infinite. Negation is
generalized into \emph{endomorphism} \(f\) whose rules involves changing
roles \(R\), that corresponds to our intuition of changing sides.
Connectives \(\land\) and \(\lor\) are generalized into ultrafilters
\(\uf\), where \(\uf\) are interpreted based on roles, as illustrated in
rule \(\lor_1\), \(\lor_2\), and \(\land\) where the rules are named
after the connectives' intended meanings. Similarly, quantifiers
\(\forall\) and \(\exists\) are generalized into ultrafilters, this time
written as \(\forallu\) to avoid conflicts.

\input{mrl}

Among all the results, the most important ones are the admissibility of
the followings. Note that \rulename{2-cut-residual} is by far the most
general form of cut, which cannot be formulated in (traditional)
classical logic. In this rule, a cut may not be a complete cut in that
it leaves a residual \(i\)-formula whose roles are the intersection of
the original cut formulas. In LK/CLL with only two roles available, the
intersection is \emph{always empty}, thus making this the same as a
regular cut. Multirole reveals this subtlety hidden in plain sight!

\begin{lemma}[Admissibility of Cut]
The following rules are admissible.

\begin{adjustbox}{minipage=1.2\textwidth,max width=1\textwidth}
\begin{gather*}
\prftree[r]{\rulename{1-cut}}
    {\vdash\Gamma,\iform{A}{\varnothing}}
    {\vdash\Gamma}
\quad
\prftree[r]{\rulename{2-cut-residual}}
    {\overline{R}_1\cap\overline{R}_2=\varnothing}
    {\vdash\Gamma,\iform{A}{R_1}}
    {\vdash\Delta,\iform{A}{R_2}}
    {\vdash\Gamma,\Delta,\iform{A}{R_1\cap R_2}}
\\
\prftree[r]{\rulename{mp-cut}}
    {\overline{R}_1\uplus\dotsb\uplus\overline{R}_n=\fullset}
    {\vdash\Gamma_1,\iform{A}{R_1}}
    {\cdots}
    {\vdash\Gamma_n,\iform{A}{R_n}}
    {\vdash\Gamma_1,\dotsc,\Gamma_n}
\end{gather*}
\end{adjustbox}
\end{lemma}

We omit the discussion of linear multirole logic (LMRL). In MRL,
formulas (including connectives) are ``global,'' \(i\)-formulas are
``local,'' and inference rules interpret connectives locally. As
mentioned in the beginning, session types are global, endpoint types are
local, and session APIs interpret global session types locally. The
design of \(\Ldeppi\) comes from this insights of MRL/LMRL.

\hypertarget{multiparty-dependent-session-types}{%
\section{Multiparty Dependent Session
Types}\label{multiparty-dependent-session-types}}

In this section, we first introduce the syntax and semantics of
\(\Ldeppi\), mention some extensions and examples, then prove its
soundness via deadlock-freeness reducibility.

\ptitle{Syntax and Static Semantics} The syntax is listed in
\cref{fig:mpdepsyntax}. We add \(\ststype\) to the statics as a new base
sorts. Static terms of sort \(\ststype\) are session types. We add
\(\stset\) as a new base sort for static integer sets to represent the
roles of a party. We use \(\stermm{\varnothing}\) for empty set and
\(\stermm{\fullset}\) for full set. We assume the existence of static
constant functions for basic set operations, e.g. \(\stermm{\uplus}\)
for disjoint union, \(\overline{\,\cdot\,}\) for complement w.r.t.
\(\stermm{\fullset}\), and \(\setminus\) for set minus.

\input{mpdepsyntax}

We use \(\spi\) for session types. We add session type constructors
\(\nil\), \(\msg\), and \(\quan\). Their sc-sorts are also given in the
signature \(\sig\). Again, \(\spi\) describe protocols \emph{globally}.
\(\stypem({\nil(r)})\) means party \(\stermm{r}\) is to close the
channel, while other parties will wait.
\(\stypem({\msg(r, \tau)\seqs\pi})\) means party \(\stermm{r}\) is to
\emph{broadcast} a value of non-linear type \(\typem(\tau)\), then
proceed according to \(\stypem(\pi)\), while others are to receive a
value of that type, then proceed according to \(\stypem(\pi)\). We
overload the name \(\msg\) and use
\(\stypem({\msg(r_1,r_2,\hat\tau)\seqs\pi})\) for sending
\emph{point-to-point linear messages}, from \(\stermm{r_1}\) to
\(\stermm{r_2}\), while others will skip the message and continue the
session following \(\spi\). \(\quan\) has sort scheme
\(\sortm({(\stint,\sigma\rightarrow\ststype)\Rightarrow\ststype})\). It
takes a role, and a \emph{static function} denoting the binder of the
quantification from sort \(\sortm(\sigma)\) to a session type, to
construct a session type that represents a \emph{global quantifier} in
the session type. The first argument denotes a role of a party who will
interpret the quantification as universal, while the others will
interpret the quantification as existential as in \cref{eg:array}.

Linear base type constructor \(\stermt{chan}\) are for endpoint types.
Given a set \(\rs\) representing the roles played by this endpoint, and
a session type \(\spi\), \(\chan{\rs}{\pi}\) is the linear type we
assign to the endpoint of roles \(\rs\) in the session \(\spi\). In LMRL
syntax, \(\chan{\rs}{\pi}\) is an analogy to \(\iform{\spi}{\rs}\) with
a formula \(\pi\) and a role set \(\rs\). Note that we inexplicit assume
some underlying full set \(\stermm{\fullset}\) and omit it for brevity.

Endpoints are resources. We use \(\dtermm{c}\) for channels, and
\(\dtermm{c^\rs}\) for an endpoint of \(\dtermm{c}\) with roles \(\rs\).
We classify \(\dtermm{c^\rs}\) as \(\dcr\). To facilitate presentation,
we define the followings.

\begin{definition}
Given a multiset of resources \(R\), i.e.~the result of \(\rho(\cdot)\)
on some term, we define the following functions.
\label[def]{def:chaneprole} \begin{align*}
&\text{All channels in $R$} &
\channels(R)          &\defeq \{c   \mid  c^\rs\in R \} \tag*{set}      \\
&\text{All endpoints in $R$} &
\endpoints(R)         &\defeq \{c^\rs \mid c^\rs\in R \} \tag*{multiset} \\
&\text{All endpoints of $\dtermm{c}$ in $R$} &
\endpoints(R,c)       &\defeq \{c_0^\rs \mid c_0^\rs\in R, c_0=c \} \tag*{multiset}  
\end{align*}
\end{definition}

We also write \(\endpoints(c)\) to mean the set of all endpoints of
channel \(\dtermm{c}\) where the disjoint union of their roles is the
underlying full set for this session. Note that if \(\pconsistent(R)\),
then all of these functions result in \emph{sets}. From now on we simply
assume that they are sets.

\begin{definition}[Consistency of Resources]
\label[def]{def:mpdepconsistent} Given a multiset of resources \(R\), we
define the consistency of \(R\) as,
\begin{itemize}
\tightlist
\item
  \(\pconsistent(\varnothing)\)
\item
  \(\pconsistent(R\uplus\endpoints(\dtermm{c}))\) iff
  \(\pconsistent(R)\), \(\endpoints(R,c)=\varnothing\), and roles of all
  endpoints of \(\dtermm{c}\) forms a partition of \(\fullset\).
\end{itemize}
\end{definition}

To assign linear types to endpoints, we add a new constant typing rule
\rulename{sig-chan} that says, if \(\dtermm{c}\) has session type
\(\spi\) in the signature, then its endpoint \(\dtermm{c^\rs}\) has
linear type \(\chan{\rs}{\pi}\). Note that when a protocol advances, the
signature will change accordingly. All other aspects of the static
semantics are the same as in \(\ATS\).

\ptitle{Session APIs} Session APIs provide local interpretations of
global session types. The dc-types assigned to them
(\cref{fig:mpdepsyntax}) ensure correct and coherent local
interpretations.

Session type for broadcasting, \(\stypem({\msg(r,\tau)\seqs\pi})\), is
interpreted by a pair \(\dbsend\) and \(\dbrecv\). The dc-type for
\(\dbsend\) says, given roles \(\rs\), a role \(\stermm{r}\), a session
type \(\spi\), and a type \(\typem(\tau)\), if \(\stermm{r}\in\rs\),
then given an endpoint of roles \(\rs\) whose type is
\(\chan{\rs}{\msg(r,\tau)\seqs\pi}\), and a message of type
\(\typem(\tau)\), \(\dbsend\) will broadcast the message, and return the
endpoint indexed by the \emph{continuation} of the protocol, which is
\(\spi\). This is shown in reduction rule \rulename{pr-bmsg}. Note that
the \(\stypem({\msg(r,\tau)})\) part is \emph{consumed}. The type system
mandates that the head of the protocol must be
\(\stypem({\msg(r,\tau)})\) and that this \(\stermm{r}\) in the protocl
must belongs to the roles of the endpoint, \(\stermm{r}\in\rs\). Only
when this proposition of the guraded type can be proven true, that this
function invocation is well-typed. Correspondingly, only when
\(\stermm{r}\notin\rs\) that one can invoke \(\dbrecv\) to receive such
broadcasting messages from party \(\stermm{r}\).

It can be proven that given a role \(\stermm{r}\) as in
\(\stypem({\msg(r,\tau)\seqs\spi})\), given a consistent collection of
endpoints of \(\dtermm{c}\) whose roles should form a partition of some
full set, there will be \emph{exactly one} endpoint whose roles contains
\(\stermm{r}\). Therefore, in a well-typed pool, there can be only one
thread invoking \(\dbsend\) at this point, while all other threads
connected by this channel can only invoke \(\dbrecv\). Also, since only
\(\dbsend\) and \(\dbrecv\) can deal with protocols starting with
\(\stypem({\msg(r,\tau)})\), any other session APIs invoked for this
protocol will be ill-typed. Since endpoint types are linear, one can
only invoke these functions once, and then must proceed to the
continuation of the protocol, denoted by \(\stypem(\seqs\,\spi)\). All
combined, each endpoint is guaranteed to follow the protocol strictly,
and all endpoints are guaranteed to be coherent/compatible with others
in the same session.

Session type \(\stypem({\msg(r_1,r_2,\hat\tau)\seqs\pi})\) is
interpreted by \(\dsend\), \(\drecv\), and \(\dskip\) for sending (at
party \(\stermm{r_1}\)), receiving (at party \(\stermm{r_2}\)), or
simply ignoring (at all others), as in \rulename{pr-msg}. Self-looping
messages are disallowed. \(\dskip\) is a \emph{proof function}, meaning
it simply changes the type of the endpoint and has no runtime effects.
It can be eliminated safely after type-checking. Note that since there
is only one sending party and one receiving party, it is now possible to
exchange \emph{linear} data. This means that endpoints as linear data
can be exchanged, resulting in \emph{higher-order} sessions. This is one
of the factors that makes the proof of deadlock-freeness much more
involved. Session type \(\stypem({\nil(r)})\) is interpreted by
\(\dclose\) at party \(\stermm{r}\) and \(\dwait\) at all other parties,
as in \rulename{pr-end}.

Quantification in the session types has never been treated in the
\(\lambda\)-calculus setting before. This work is the first of its kind
to support quantification and polymorphism in the session types.
\(\stypem({\quan(r,\lambda a{:}\sigma.\pi)})\) is a quantifier that
needs interpretation, just as \(\forallu\) in LMRL. It is interpreted by
\(\dunify\) at party \(\stermm{r}\) as universal quantification and
\(\dexify\) at all other parties as existential quantification, shown in
\rulename{pr-quan} and \cref{eg:array}. They are proof functions as
well. Note that if we quantify over session types, we obtain
\emph{polymorphic} sessions. See \cref{eg:poly}.

Given a thread function, \(\dfork\) creates a new thread and connects to
it with a fresh channel \(\dtermm{c}\). The dc-type of \(\dfork\)
mandates that \(\stermm{\rs_1}\) and \(\stermm{\rs_2}\) form a partition
of the full set, ensuring the consistency of channel endpoints.
\(\dcut\) is for connecting two channels of the same session type, i.e.,
forwarding. It loosely corresponds to the \rulename{2-cut-residual} rule
of LMRL. Given one endpoint from each of the two channels, \(\dcut\)
connects the two channels into a single channel and returns a
\emph{residual} endpoint of roles \(\stermm{\rs_1\cap\rs_2}\) to the
caller, also connected in this channel. The admissibility of cut in LMRL
means connecting \emph{multiple} channels of the same type via
forwarding is equivalent to establishing just a \emph{single} channel
connecting all parties from the very beginning. \(\delim\) is for
eliminating an endpoint with empty role sets. \(\dsplit\) is for
splitting an endpoint into two disjoint endpoints, in separate threads.

\ptitle{Dynamic Semantics} Term reductions model normalizations in a
single thread. In a thread pool \(\dpool\), we use pool reductions. Note
that our formulation is for synchronous communications for simplicity,
although our implementations fully support asynchronous communications.
We define partial ad-hoc redexes in \cref{fig:mpdepred}. Only matching
partial redexes can reduce according to some pool reduction rules.

\begin{definition}[Matching Partial Redexes]
\label[def]{def:mpdepmatch} For any channel \(\dtermm{c}\),
\begin{itemize}
\tightlist
\item
  \(\pmatch(\{\dtermm{\dbsend(c^{\rs_1}, v)}, \dtermm{\dbrecv(c^{\rs_2})},\dotsc,\dtermm{\dbrecv(c^{\rs_n})}\})\)
\item
  \(\pmatch(\{\dtermm{\dsend(c^{\rs_1}, v)}, \dtermm{\drecv(c^{\rs_2})},\dtermm{\dskip(c^{\rs_3})},\dotsc,\dtermm{\dskip(c^{\rs_n})}\})\)
\item
  \(\pmatch(\{\dtermm{\dclose(c^{\rs_1}, v)}, \dtermm{\dwait(c^{\rs_2})},\dotsc,\dtermm{\dwait(c^{\rs_n})}\})\)
\item
  \(\pmatch(\{\dtermm{\dunify(c^{\rs_1}, v)}, \dtermm{\dexify(c^{\rs_2})},\dotsc,\dtermm{\dexify(c^{\rs_n})}\})\)
\end{itemize}
\noindent where
\(\endpoints({\dtermm{c}})=\{\dtermm{c^{\rs_1}},\dotsc,\dtermm{c^{\rs_n}}\}\).
We write \(\pmatch(\{\dtermm{e_1},\dotsc,\dtermm{e_n}\})\) for
\(\pmatch(\dtermm{e_1'},\dotsc,\dtermm{e_n'})\) if
\(\dtermm{e_1}=\dtermm{E[e_1']},\dotsc,\dtermm{e_n}=\dtermm{E[e_n']}\).
\end{definition}

\input{mpdepred}

\ptitle{Extensions} Some common features are intentionally left out for
brevity. We mention some very briefly. Branching in the session types
can be supported by adding a new session type constructor,
\(\stypem({\branch(r,\pi_1,\pi_2)})\) and a pair of cooresponding
session APIs \(\doffer\) and \(\dchoose\). Party \(\stermm{r}\) will be
offering the choice, and all other parties need to choose. Recursive
sessions can be supported by adding
\(\stypem({\fix(\lambda a{:}\ststypec.\pi)})\). A session API
\(\drecurse\) can be added to unroll \(\chan{\rs}{\fix(f)}\) into
\(\chan{\rs}{f(\fix(f))}\) for any \(\rs\) and \(\stermm{f}\).
\(\dfork\) only forms sessions locally. One can introduce
\(\stypem({\init(r,\pi)})\) for forming sessions distributedly. A pair
of APIs, \(\daccept\) (at \(\stermm{r}\)) and \(\drequest\) (at all
others) can be provided. Note that after \(\drequest\), a new thread is
created and the new endpoint is passed to the thread. \(\drequest\)
\emph{can not} return the endpoint to the current thread as it breaks
relaxation, thus df-reducibility, and will cause a loop. See
\cite{Wu:2017wd} for more possible constructs.

\begin{example}
\label[eg]{eg:poly} We model a cloud service. When the provider \(P\)
gives the server \(S\) a function that serves \emph{any} protocol
\(\stypem(x)\) \emph{once}, the server \(S\) will \emph{repeatedly}
serving that protocol to client \(C\). As a syntactical sugar, we will
write \(\dtermm{e_1;e_2}\) as a shorthand for
\(\dtermm{\dsnd(\dpair(e_1,e_2))}\). We also write
\(\dtermm{\dlet{x}{e_1}{e_2}}\) for
\(\dtermm{\dapp(\dlam x.e_2, e_1)}\). The session type is

\begin{adjustbox}{varwidth=1.1\textwidth,max width=1\textwidth}
$$
\stypem({\quan(P,\lambda x{:}\ststypec.\msg(P,S,\chanc{S}{x}\rightarrow\sunitc)\seqs\fix(\lambda y{:}\ststypec.\msg(C,S,\chanc{S}{x})\seqs y)})
$$
\end{adjustbox}

~

\noindent This is a \emph{polymorphic}, \emph{higher-order} session
type, that involves both a 3-party session among \(S\)/\(P\)/\(C\), and
a 2-party session (of session type \(\stypem(x)\)) between \(S\)/\(C\).
The program of \(S\) can be implemented below. The annotations on the
right denote the endpoint's type \emph{after} the invocation of the
session API on that line. \break

\def\dloop{\dtermt{loop}}
\begin{adjustbox}{varwidth=1.4\textwidth,max width=1\textwidth}
$$
\begin{array}{ll}
\dtermm{\dlam c.\dlet{c}{\dexify(c)}{}}       & \dtermm{c}:\typem({\chan{S}{\msg(P,S,\chanc{S}{x}\rightarrow\sunitc)\seqs\fix(\cdots)}})\\ 
\dtermm{\dlet{\dpair(c,f)}{\drecv(c)}{}}      & \dtermm{c}:\typem({\chan{S}{\fix(\lambda y{:}\ststypec.\msg(C,S,\chanc{S}{x})\seqs y)}})\\ 
\dtermm{\dkw(let)~\dtermt{loop}=\dfix g.\dlam x.\dlet{x}{\drecurse(x)}{}} & \dtermm{x}:\typem({\chan{S}{\msg(C,S,\chanc{S}{x})\seqs\fix(\cdots)}})\\ 
\dtermm{\dlet{\dpair(x,y)}{\drecv(x)}{\dapp(f,y);\dapp(g,x)}}       &\dtermm{y}:\typem({\chan{S}{x}})\\
\dtermm{\dkw(in)~\dapp(\dtermt{loop}, c)}
\end{array}
$$
\end{adjustbox}
\end{example}

On the first line, \(\dexify\) interprets \(\quan\) as existential, and
then immediately \rulename{ty-exists-elim} (\cref{fig:mtlcalltyping}) is
used to eliminate the quantifier. On the second line, \(S\) receives the
thread function that \(P\) wants to let \(S\) repeatedly serve. The
function can serve protocol \(\stypem(x)\), \emph{for any}
\(\stypem(x)\). On the third and fouth line, we define a
\emph{recursive} function (a feature in \(\ATS\) not covered here)
called \(\dtermt{loop}\). The function will, on the third line, unroll
the endpoint once, and on the forth line, call \(\drecv\) to receive an
endpoint of type \(\chan{S}{x}\) created by and sent from \(C\),
followed by invoking \(P\)-supplied function \(\dtermm{f}\) on
\(C\)-supplied enpoint \(\dtermm{y}\), providing service \(\stypem(x)\).
On the fifth line, \(S\) invokes the \(\dtermt{loop}\) function with
endpoint \(\dtermm{c}\).

\begin{wrapfigure}{R}{0.4\textwidth}
\begin{adjustbox}{minipage=0.6\textwidth,max width=0.4\textwidth}
$$
\begin{tikzcd}[column sep=normal]
(1)              & S \arrow[dash,dl] & (2)             & S\arrow[dash,dl] & \\
P \arrow[dash,r] & C \arrow[dash,u]  & P\arrow[dash,r] & C \arrow[dash,u] \arrow[r,dashed,"\text{fork}"] & C' \arrow[equal,ul,swap,"x"] 
\end{tikzcd}
$$
\end{adjustbox}
\end{wrapfigure}

The network topology is shown here. In (1), \(P\)/\(C\)/\(S\) is
connected in a 3-party session described above. Then client \(C\)
invokes \(\dfork\), spawning a new thread with a new channel, sending
one endpoint to \(S\), and starts a new 2-party session with protocol
\(\stypem(x)\) between \(S\) and a child thread of \(C\), denoted as
\(C'\).

This example can be successfully type-checked in our current
implementation of \(\Ldeppi\).

\hypertarget{deadlock-freeness-reducibility}{%
\subsection{Deadlock-freeness
Reducibility}\label{deadlock-freeness-reducibility}}

The technique of df-reducibility is introduced in our early work
\cite{Xi:2016ue} for \emph{binary} session types. It is adapted for
multiparty session types. The notion captures the invariance of pool
reduction that, at any time, there are no loops or self-loops in all
endpoint connections. The proof of \cref{lemma:mtlcsimpledeadlockfree}
carried out as follows. First, we show that reduction preserves
df-reducibility (\cref{lemma:mtlcsimplereducibility}). Then
df-reducibility implies relaxation (\cref{prop:deadlock}) which in turn
means reducible (\cref{prop:pigeonhole}). Notably, relaxation
(\cref{def:relaxed}) is \emph{not} an invariant during reduction, e.g.,
for the case of \rulename{pr-end}, and that is why we strengthened it
and formalized df-reducibility (\cref{def:dfreducible}). This proof
technique, and particularly \cref{prop:deadlock}, guided the choices of
when to return an endpoint to the caller and when to spawn a new thread,
like \(\daccept\)/\(\drequest\) discussed before.

\begin{definition}[Abstract Collections of Endpoints]
We use \(M\) to denote a finite \emph{set} of endpoints and \(\M\) to
denote a finite \emph{set} of \(M\) where all \(M\in\M\) are pair-wise
disjoint. We use \(\bigcup\M\) to mean the (disjoint) union of all
\(M\in\M\). For any channel \(\dtermm{c}\), we assume either
\(\endpoints(\dtermm{c})\subset\bigcup\M\) or
\(\endpoints(\dtermm{c})\cap\bigcup\M=\varnothing\).
\end{definition}

We lift the definition of \(\endpoints(\cdot)\) and \(\channels(\cdot)\)
onto \(M\) and \(\M\) to collect endpoints and channels from them.

\begin{definition}[Deadlock-free Reduction $\dfreduce$]
\label[def]{def:dfreduce} We write \(\M \dfreducevia(c) \M'\) if there
exists a channel \(\dtermm{c}\), some sets of endpoints \(M_i\) s.t.
\begin{gather*}
\endpoints(\dtermm{c})=\{\dtermm{c^{\rs_1}},\dotsc,\dtermm{c^{\rs_n}}\} \text{  and  } \dtermm{c^{\rs_i}} \in M_i \in \M \text{  for $1\leq i \leq n$}\\
\M'=(\M\setminus\{M_1,\dotsc,M_n\})\cup \{M'\} \text{  and  } M'=(M_1\cup\dotsb\cup M_n)\setminus\endpoints(\dtermm{c})
\end{gather*} We also write \(\M\dfreduce\M'\) if there exists some
\(\dtermm{c}\) s.t. \(\M\dfreducevia(c)\M'\). We say \(\M\) is df-normal
if it can not be further df-reduced, denoted as \(\M\dfnormal\).
\end{definition}

\begin{definition}[Deadlock-freeness Reducibility]
\label[def]{def:dfreducible} We write \(\pdfreducible(\M)\) if
\begin{itemize}
\tightlist
\item
  each \(M\in\M\) is an empty set, or
\item
  \(\M\) is not df-normal, and for any \(\M'\) where \(\M\dfreduce\M'\)
  holds, \(\pdfreducible(\M')\).
\end{itemize}
\end{definition}

There are certain properties about df-reducibility that are easy to
prove. We omit them and refer readers to \cite{Xi:2016ue}. We only note
that empty elements can always be removed from \(\M\) without breaking
df-reducibility.

\begin{definition}[Relaxed]
\label[def]{def:relaxed} Let \(|\M|\) be the number of \emph{non-empty}
elements. \[
\prelaxed(\M)\defeq 
\begin{cases}
|\M| \geq |\endpoints(\M)| - |\channels(\M)| + 1 \\
|\M| = 0
\end{cases}
\]
\end{definition}

\begin{proposition}[Pigeonhole]
\label[prop]{prop:pigeonhole} If \(\prelaxed(\M)\) and \(|\M|>0\) where
\(|\M|\) is the number of \emph{non-empty} elements, then
\(\M\dfreducevia(c)\M'\) for some \(\dtermm{c}\) and \(\M'\). Namely,
\(\M\) is not df-normal.
\end{proposition}

\begin{proposition}
\label[prop]{prop:deadlock} \(\neg\prelaxed(\M)\) implies
\(\neg\pdfreducible(\M)\).
\end{proposition}

We now make \(M\) and \(\M\) concrete. Let
\(\pconsistent(\rho(\dpool))\), we define \(M(\dtermm{\dpool(t)})\) as
\(\endpoints(\rho(\dtermm{\dpool(t)}))\), \(\M(\dpool)\) as
\(\bigcup_t\, \{M(\dtermm{\dpool(t)})\}\), \(\prelaxed(\dpool)\) as
\(\prelaxed(\M(\dpool))\), and \(\pdfreducible(\dpool)\) as
\(\pdfreducible(\M(\dpool))\). We also define \(\pblocked(\dtermm{e})\)
as \(\dtermm{e} = \dtermm{E[e']}\) for some evaluation context \(E\) and
partial redex \(\dtermm{e'}\). A blocked expression is blocked on some
endpoint \(\dtermm{c^r}\) of some channel \(\dtermm{c}\). We write
\(\pblocked(\dtermm{e},\dtermm{c^r})\) or
\(\pblocked(\dtermm{e},\dtermm{c})\) to make it explicit.

\begin{lemma}[Reducibility of Pools]
\label[lemma]{lemma:mtlcsimplereducibility} Pool reduction preserves
df-reducibility. \(\pconsistent(\rho(\dpool))\),
\(\pdfreducible(\dpool)\), and \(\dpool\creduce\dtermm{\dpool'}\)
implies \(\pdfreducible(\dtermm{\dpool'})\).
\end{lemma}

\begin{lemma}[Deadlock-free]
\label[lemma]{lemma:mtlcsimpledeadlockfree} Let \(\dpool\) be a
well-typed pool s.t. \(\dtermm{\dpool(0)}\) is either a value
\(\dtermm{v}\) without endpoints or \(\pblocked(\dtermm{\dpool(0)})\),
and \(\pblocked(\dtermm{\dpool(t)})\) for \(0<t\in\dom(\dpool)\). If
\(\dpool\) is obtained from evaluating an initial pool without any
channels, then there exist \(t_1,\dotsc,t_n\in\dom(\dpool)\) s.t.
\(\pmatch(\{\dtermm{\dpool(t_1)},\dotsc, \dtermm{\dpool(t_n)}\})\).
\end{lemma}

\begin{proof}
The initial pool is consistent by \cref{def:mpdepconsistent} and
df-reducible by \cref{def:dfreducible} since it contains no endpoints.
Therefore \(\pdfreducible(\dpool)\) by
\cref{lemma:mtlcsimplereducibility} and \(\pconsistent(\dpool)\) by
\cref{thm:mpdepsubjectreduction}. By \cref{prop:deadlock} we have
\(\prelaxed(\dpool)\). Parallel to \cref{prop:pigeonhole}, by the
Pigeonhole Principle, there exist \(t_1,\dotsb,t_n\in\dom(\dpool)\) s.t.
\(\pblocked(\dtermm{\dpool(t_1)},\dtermm{c^{\rs_1}})\),\(\dotsb\), and
\(\pblocked(\dtermm{\dpool(t_n)}, \dtermm{c^{\rs_n}})\) for some channel
\(\dtermm{c}\) where
\(\endpoints(\dtermm{c})=\{\dtermm{c^{\rs_1}},\dotsc,\dtermm{c^{\rs_n}}\}\).
Since \(\dpool\) is well-typed and consistent, these endpoints are
assigned coherent types by rule \rulename{sig-chan}. Therefore we have
\(\pmatch(\{\dtermm{\dpool(t_1)},\dotsb,\dtermm{\dpool(t_n)}\})\) by
straightforwardly examining the typing derivations of the partial
redexes.
\end{proof}

\hypertarget{soundness}{%
\subsection{Soundness}\label{soundness}}

\begin{theorem}[Subject Reduction]
\label[thm]{thm:mpdepsubjectreduction} Assume
\(\vdash\dtermm{\dpool_1}:\typem(\hat\tau)\) under some signature
\(\sig_1\), \(\pconsistent(\rho(\dtermm{\dpool_1}))\), and
\(\dtermm{\dpool_1}\creduce\dtermm{\dpool_2}\) for some
\(\dtermm{\dpool_2}\). Then \(\vdash\dtermm{\dpool_2}:\typem(\hat\tau)\)
under a cooresponding signature \(\sig_2\) and
\(\pconsistent(\rho(\dtermm{\dpool_2}))\).
\end{theorem}

\begin{proof}
Proof by induction on the derivation of
\(\dtermm{\dpool_1}\creduce\dtermm{\dpool_2}\).
\end{proof}

\begin{theorem}[Progress]
\label[thm]{thm:mpdepprogress} Assume \(\vdash\dpool:\typem(\hat\tau)\)
and \(\pconsistent(\dpool)\). We have the following possibilities:
\begin{itemize}
\tightlist
\item
  \(\dpool\) is a singleton mapping \(\dtermm{0\mapsto v}\).
\item
  \(\dpool\creduce\dtermm{\dpool'}\) holds for some
  \(\dtermm{\dpool'}\).
\end{itemize}
\end{theorem}

\begin{proof}
In the case where all threads are blocked expressions,
\cref{lemma:mtlcsimpledeadlockfree} is needed.
\end{proof}

\hypertarget{related-works-and-conclusions}{%
\section{Related Works and
Conclusions}\label{related-works-and-conclusions}}

To address a problem found in higher-order sessions that breaks type
preservation theorem \cite{Yoshida:2007hk}, \emph{polarity} and
\emph{balanced typing} \cite{Gay:2005gc,Lindley:2015cr} are introduced
to distinguish the two ends of a channel syntactically, and to ensure
their typing duality, respectively. The advantage of our formulation is
that endpoint types are inherently balanced by \rulename{sig-chan},
given that \(\sig\) only records global session types, and that all
roles of endpoints of a channel should form a partition of \(\fullset\)
due to \cref{def:mpdepconsistent}. Also, we found our approach much
cleaner for generalizing into multiparty session types compared to the
polarized approach. \cite{Honda:2008hi,Bejleri:2009eo,Yoshida:2010ht}
explored multiparty session types, \cite{Toninho:2011in,Pfenning:2011ce}
explored binary dependent session types. However, the present work is
the first to combine dependent types with multiparty session types.
\cite{Abramsky:1993io,Abramsky:1994ez,Bellin:1994ua} explored the
connection of process calculi with linear logic.
\cite{Caires:2010gi,Wadler:2012ua} established a Curry-Howard
correspondence between session typed process calculus with propositions
in linear logic. Later works
\cite{Carbone:2015hl,Carbone:2016kd,Carbone:2017ki} developed a
generalized notion of duality called coherence, to correspond multiparty
session types with propositions in linear logic equipped with a separate
proof system for coherence. We consider their formulation as an
extension instead of a generalization, since the coherence rule is a
separate proof system, and the well-known duality of the axiom rule and
the cut rule is lost. We consider our work as a formal generalization,
with LK/CLL being a special case of MRL/LMRL. Also, results like
\rulename{2-cut-residual} are made possible only via multirole.

We have demonstrated the first formulation of multiparty dependent
session types, the df-reducibility proof technique, and the intuitions
behind multirole logic. We point out that by representing session types
as program terms, by implementing propositions in the guarded types of
the APIs as runtime assertions, one can still greatly benefit from the
ability to deterministically know in advance about whether the system is
deadlock-free, \emph{even in a language without dependent types and
linear types}. This is precisely the significance of our practical
system. With the help of such formal reasoning, concurrency can be made
better, safer, and more accessible.
\bibliographystyle{splncs04}
\bibliography{library}
%




\end{document}

%% file: mtlcallsyntax.tex
\begin{figure}
Syntax of Statics
\[
\begin{array}{rrcl}
\text{base sorts}       \quad & \sortm(b)      & \defeq & \sortt(int) \mid \sortt(bool)                          \mid \sortt(type)                              \mid \sortt(vtype)      \\
\text{sorts}            \quad & \sortm(\sigma) & \defeq & \sortm(b)   \mid \sortm(\sigma_1 \rightarrow \sigma_2) \\
\text{static constants} \quad & \scx           & \defeq & \scc        \mid \scf                                  \\
\text{static terms}     \quad & \stermm{s}     & \defeq & \stermm{a}  \mid \stermm{\scx(\vv{s})}                 \mid \stermm{\lambda a {:} \sigma . s} \mid \stermm{s_1 (s_2)} 
\end{array}
\]

Syntax of Types
\[
\begin{array}{rrcl}
\text{types}        \quad & \typem(\tau)     & \defeq & \typem(a) \mid \typem(\delta(\vv{s}))     \mid \sunit                    \mid \typem(\tau_1\times\tau_2)            \mid \typem(\tau_1\rightarrow\tau_2)
                                                                  \mid \typem(P\sguard\tau)       \mid \typem(P\sassert\tau)     \mid \typem({\forall a{:}\sigma.\tau})     \mid \typem({\exists a{:}\sigma.\tau})     \\
\text{linear types} \quad & \typem(\hat\tau) & \defeq & \typem(a) \mid \typem(\hat\delta(\vv{s})) \mid \typem(\tau)              \mid \typem(\hat\tau_1\otimes\hat\tau_2)   \mid \typem(\hat\tau_1\multimap\hat\tau_2)
                                                                  \mid \typem(P\sguard\hat\tau)   \mid \typem(P\sassert\hat\tau) \mid \typem({\forall a{:}\sigma.\hat\tau}) \mid \typem({\exists a{:}\sigma.\hat\tau})
\end{array}
\]

Syntax of Dynamics
\[ %
\begin{array}{rrcl}
\text{constants} \quad & \dcx       & \defeq & \dcc                  \mid \dcf                     \\
\text{terms}     \quad & \dtermm{e} & \defeq & \dtermm{x}            \mid \dtermm{\dcx(\vv{e})}    \mid \dcr                  \mid \dtermm{\dlam x.e}           \mid \dtermm{\dapp(e_1,e_2)}                   \mid                    \\
                       &            &        & \dunit                \mid \dtermm{\dpair(e_1,e_2)} \mid \dtermm{\dfst(e)}     \mid \dtermm{\dsnd(e)}            \mid \dtermm{\dlet{\dpair(x_1,x_2)}{e_1}{e_2}} \mid                    \\
                       &            &        & \dtermm{\dguardi(v)}  \mid \dtermm{\dguarde(e)}     \mid \dtermm{\dassert(e)}  \mid \dlet{\dassert(x)}{e_1}{e_2} \mid \\
                       &            &        & \dtermm{\dforalli(v)} \mid \dtermm{\dforalle(e)}    \mid \dtermm{\dexists(e)}  \mid \dlet{\dexists(x)}{e_1}{e_2} \\
\text{values}    \quad & \dtermm{v} & \defeq & \dtermm{x}            \mid \dcr                     \mid \dtermm{\dcc(\vv{v})} \mid \dunit                       \mid \dtermm{\dpair(v_1,v_2)}                  \mid \dtermm{\dlam x.e} \mid 
                                               \dtermm{\dguardi(v)}  \mid \dtermm{\dassert(v)}     \mid \dtermm{\dforalli(v)} \mid \dtermm{\dexists(v)}         \\
\text{pools}     \quad & \dpool     & \defeq & \dtermm{\varnothing}  \mid \dtermm{\dpool,t{:}e}
\end{array}
\]

Some Signatures of Dynamic Constants
\[
\begin{array}{rcl}
\dfork  & : & \typem({(\sunit\multimap\sunit)\Rightarrow\sunit}) 
\end{array}
\]

\caption{Selected Syntax of $\ATS$}
\label{fig:mtlcallsyntax}
\end{figure}

%% file: mtlcalltyping.tex
\begin{figure}

\begin{adjustbox}{varwidth=1.3\textwidth,max width=\textwidth}
\input{mtlcdeptyping}
\end{adjustbox}
\caption{Selected Typing Rules of $\ATS$}
\label{fig:mtlcalltyping}
\end{figure}

%% file: mtlcdeptyping.tex
$\boxed{\Sigma;\stermm{\vv P};\Gamma;\Delta\vdash\dtermm{e}:\typem(\hat\tau)}$
\begin{gather*}
\prftree[r]{\rulename{ty-res}}
	{\sig\vDash\dcr:\typem(\hat\delta(\vv s))}
	{\Sigma;\stermm{\vv P};\Gamma;\varnothing\vdash\dcr:\typem(\hat\delta(\vv s))}
\quad
\prftree[r]{\rulename{ty-pool}}
	{ \vdash \dtermm{\dpool(0)}:\typem(\hat\tau)}
	{ \vdash \dtermm{\dpool(i)}:\sunit}
	{\text{for $0<i\in\dom(\dtermm{\dpool})$}}
	{ \vdash \dpool:\typem(\hat\tau)}
\\
\prftree[r]{\rulename{ty-guard-intro}}
	{\Sigma;\stermm{\vv P},\stermm{P_0};\Gamma;\Delta\vdash\dtermm{v}:\typem(\hat\tau)}
	{\Sigma;\stermm{\vv P};\Gamma;\Delta\vdash\dtermm{\dguardi(v)}:\typem(P_0\sguard\hat\tau)}
\quad
\prftree[r]{\rulename{ty-guard-elim}}
	{\Sigma;\stermm{\vv P};\Gamma;\Delta\vdash\dtermm{e}:\typem(P_0\sguard\hat\tau)}
	{\Sigma;\stermm{\vv P}\vDash\stermm{P_0}}
	{\Sigma;\stermm{\vv P};\Gamma;\Delta\vdash\dtermm{\dguarde(e)}:\typem(\hat\tau)}
\\
\prftree[r]{\rulename{ty-forall-intro}}
	{\Sigma,\stermm{a}:\sortm(\sigma);\stermm{\vv P};\Gamma;\Delta\vdash\dtermm{v}:\typem(\hat\tau)}
	{\Sigma;\stermm{\vv P};\Gamma;\Delta\vdash\dtermm{\dforalli(v)}:\typem(\forall a{:}\sigma.\hat\tau)}
\quad
\prftree[r]{\rulename{ty-forall-elim}}
	{\Sigma;\stermm{\vv P};\Gamma;\Delta\vdash\dtermm{e}:\typem(\forall a{:}\sigma.\hat\tau)}
	{\Sigma\vdash\stermm{s}:\sortm(\sigma)}
	{\Sigma;\stermm{\vv P};\Gamma;\Delta\vdash\dtermm{\dforalle(e)}:\typem(\hat\tau[a\mapsto s])}
\\
\prftree[r]{\rulename{ty-exists-intro}}
	{\Sigma;\stermm{\vv P};\Gamma;\Delta\vdash\dtermm{e}:\typem(\hat\tau[a \mapsto s])}
	{\Sigma\vdash\stermm{s}:\sortm(\sigma)}
	{\Sigma;\stermm{\vv P};\Gamma;\Delta\vdash\dtermm{\dexists(e)}:\typem(\exists a{:}\sigma.\hat\tau)}
\\
\prftree[r]{\rulename{ty-exists-elim}}
	{\Sigma;\stermm{\vv P};\Gamma;\Delta\vdash\dtermm{e_1}:\typem(\exists a{:}\sigma.\hat\tau_1)}
	{\Sigma,\stermm{a}:\sortm(\sigma);\stermm{\vv P};(\Gamma;\Delta),\dtermm{x}:\typem(\hat\tau_1)\vdash\dtermm{e_2}:\typem(\hat\tau_2)}
	{\Sigma;\stermm{\vv P};\Gamma;\Delta\vdash\dtermm{\dlet{\dexists(x)}{e_1}{e_2}}:\typem(\hat\tau_2)}
\end{gather*}

%% file: mtlcallrho.tex
\begin{figure}
\[
\begin{array}{rclrclrclrcl}
\rho(\dcr)                                           & = & \{\dcr\}                                               &
\rho(\dtermm{x})                                     & = & \varnothing                                            &
\rho(\dunit)                                         & = & \varnothing                                            &
\rho(\dtermm{\dpair(e_1, e_2)})                      & = & \rho(\dtermm{e_1})\uplus\rho(\dtermm{e_2})             \\
\rho(\dtermm{\dfst(e)})                              & = & \rho(\dtermm{e})                                       &
\rho(\dtermm{\dsnd(e)})                              & = & \rho(\dtermm{e})                                       &
\rho(\dtermm{\dlam x.e})                             & = & \rho(\dtermm{e})                                       &
\rho(\dtermm{\dapp(e_1,e_2)})                        & = & \rho(\dtermm{e_1})\uplus\rho(\dtermm{e_2})             
\end{array}
\]
\caption{Selected Definition of $\rho(\cdot)$}
\label{fig:mtlcallrho}
\end{figure}

%% file: mtlcallred.tex
\begin{figure}
Redexes
\[
\begin{array}{rcl}
\text{pure redex}   & \defeq & \dtermm{\dapp(\dlam x.e, v)}                \mid \dtermm{\dlet{\dpair(x_1,x_2)}{\dpair(v_1,v_2)}{e}} \mid \dtermm{\dfst(\dpair(v_1,v_2))}             \mid                                  \\
                    &        & \dtermm{\dsnd(\dpair(v_1,v_2))}             \mid \dtermm{\dguarde(\dguardi(v))}                      \mid \dtermm{\dlet{\dassert(x)}{\dassert(v)}{e}} \mid \dtermm{\dforalle(\dforalli(v))} \mid \\
                    &        & \dtermm{\dlet{\dexists(x)}{\dexists(v)}{e}} \\
\text{ad-hoc redex} & \defeq & \dtermm{\dcf(\vv v)}
\end{array}
\]

Contractums
\[
\begin{array}{lrcll}
\text{(\rulename{tr-beta})}   & \dtermm{\dapp(\dlam x.e, v)}                        & \vreduce & \dtermm{e\map(x,v)}                 \\
\text{(\rulename{tr-guard})}  & \dtermm{\dguarde(\dguardi(v))}                      & \vreduce & \dtermm{v}                          \\

\text{(\rulename{tr-adhoc})} & \dtermm{\dcf(\vv{v})} & \vreduce & \dtermm{v'} \quad\text{if $\dcf$ is defined at $\dtermm{\vv{v}}$ and the result is $\dtermm{v'}$} \\
&&\vdots&
\end{array}
\]


Reductions
\[
\begin{array}{rcll}
\dtermm{E[e]} & \reduce & \dtermm{E[e']} & \text{if $\dtermm{e} \vreduce \dtermm{e'}$} 
\end{array}
\]

Pool Reductions
\[
\begin{array}{lrcll}
\text{(\rulename{pr-fork})} & \dtermm{\dpool,t{:} E[\dfork(\dlam x.e)]} & \creduce & \multicolumn{2}{l}{\dtermm{\dpool,t{:} E[\dunit], t'{:} e\map(x,\dunit)}}             \\
\text{(\rulename{pr-gc})}   & \dtermm{\dpool,t{:}\dunit}                & \creduce & \dpool                      &  \text{if $t > 0$}                          \\
\text{(\rulename{pr-lift})} & \dtermm{\dpool,t{:} e}                    & \creduce & \dtermm{\dpool,t{:} e'} &  \text{if $\dtermm{e} \reduce \dtermm{e'}$}
\end{array}
\]

\caption{Selected Reductions in $\ATS$}
\label{fig:mtlcallred}
\end{figure}

%% file: mrlintuition.tex
\begin{figure}
\begin{adjustbox}{varwidth=1.2\textwidth,max width=\textwidth}
\begin{gather*}
\tag{Two-sided}
\vcenter{\prftree[r]{\rulename{id}}
    {\ }
	{A \vdash A}}
\quad
\vcenter{\prftree[r]{\rulename{cut}}
	{\Gamma \vdash A, \Delta}
	{\Gamma', A \vdash \Delta'}
	{\Gamma,\Gamma' \vdash \Delta, \Delta'}}
\quad
\vcenter{\prftree[r]{\rulename{$\neg$L}}
	{\Gamma\vdash A, \Delta}
	{\Gamma, \neg A\vdash \Delta}}
\quad
\vcenter{\prftree[r]{\rulename{$\neg$R}}
	{\Gamma, A \vdash \Delta}
	{\Gamma \vdash \neg A, \Delta} }
\\
\tag{One-sided}
\vcenter{\prftree[r]{\rulename{id}}
	{\ }
	{\vdash A,\neg A} }
\quad
\vcenter{\prftree[r]{\rulename{cut}}
	{\vdash\Gamma,A}
	{\vdash\Delta,\neg A}
	{\vdash\Gamma,\Delta}}
\\
\tag{Many-sided}
\vcenter{\prftree[r]{\rulename{id-two-sided}}{\ }
    {\iform{A}{1}\vdash\iform{A}{0}}}
\quad 
\vcenter{\prftree[r]{\rulename{id-one-sided}}{\ }
    {\vdash\iform{A}{0},\iform{\neg A}{0}}}
\quad
\vcenter{\prftree[r]{\rulename{id-two-sided-on-one-side}}{\ }
    {\vdash\iform{A}{0},\iform{A}{1}}}
\end{gather*}
\end{adjustbox}
\caption{Intuitions of MRL}
\label{fig:mrlintuition}
\end{figure}

%% file: mrl.tex
\def\mrlwk(#1,#2){
\prftree[r]{\rulename{w}}
	{\vdash\Gamma}
	{\vdash\Gamma,\iform{#1}{#2}}
}
\def\mrlctr(#1,#2){
\prftree[r]{\rulename{c}}
	{\vdash\Gamma,\iform{#1}{#2},\iform{#1}{#2}}
	{\vdash\Gamma,\iform{#1}{#2}}
}

\def\mrlid(#1){
\prftree[r]{\rulename{id}}
	{R_1\uplus\dotsb\uplus R_n=\fullset}
	{\vdash\iform{#1}{R_1},\dotsc,\iform{#1}{R_n}}
}

\def\mrlneg(#1,#2,#3){
\prftree[r]{$\neg$}
	{\vdash\Gamma,\iform{#1}{\inv{#2}(#3)}}
	{\vdash\Gamma,\iform{#2(#1)}{#3}}
}

\def\mrlconj(#1,#2,#3){
\prftree[r]{$\land$}
	{#3\in\uf}
	{\vdash\Gamma,\iform{#1}{#3}}
	{\vdash\Gamma,\iform{#2}{#3}}
	{\vdash\Gamma,\iform{#1\conju#2}{#3}}
}

\def\mrldisjl(#1,#2,#3){
\prftree[r]{$\lor_1$}
	{#3\notin\uf}
	{\vdash\Gamma,\iform{#1}{#3}}
	{\vdash\Gamma,\iform{#1\conju#2}{#3}}
}

\def\mrldisjr(#1,#2,#3){
\prftree[r]{$\lor_2$}
	{#3\notin\uf}
	{\vdash\Gamma,\iform{#2}{#3}}
	{\vdash\Gamma,\iform{#1\conju#2}{#3}}
}

\def\mrlexists(#1,#2){
\prftree[r]{$\exists$}
	{#2\notin\uf}
	{\vdash\Gamma,\iform{#1\subst{x}{t}}{#2}}
	{\vdash\Gamma,\iform{\forallu x.#1}{#2}}
}

\def\mrlforall(#1,#2){
\prftree[r]{$\forall$}
	{#2\in\uf}
	{x\notin\Gamma}
	{\vdash\Gamma,\iform{#1}{#2}}
	{\vdash\Gamma,\iform{\forallu x.#1}{#2}}
}

\begin{figure}
\begin{adjustbox}{varwidth=1\textwidth,max width=\textwidth}
\[
\begin{array}{rlcl}
\text{Formulas}    & && A,B \defeq a \mid f(A) \mid A \conju B \mid \forallu x.A \\
\text{$i$-formulas}& && \iform{A}{R}
\end{array}
\]

\begin{gather*}
\mrlid(a)\quad
\mrlwk(A,R)\quad\mrlctr(A,R)\quad
\mrlneg(A,f,R) \\
\mrldisjl(A,B,R) \quad \mrldisjr(A,B,R) \quad
\mrlconj(A,B,R) \\
\mrlexists(A,R) \quad
\mrlforall(A,R)
\end{gather*}
\end{adjustbox}
\caption{Multirole Logic}
\label{fig:mrl}
\end{figure}

%% file: mpdepsyntax.tex
\begin{figure}
\begin{adjustbox}{varwidth=1.2\textwidth,max width=1\textwidth}
Additional Syntax of Statics
\[
\begin{array}{rrcl}
\text{base sorts} \quad & \sortm(b)   & \defeq & \cdots \mid \sortt(set) \mid \ststype   \\
\text{roles}      \quad & \stermm{r}  & \defeq & \cdots \mid \stermm{-1} \mid \stermm{0} \mid \stermm{1}  \mid \cdots \\
\text{role sets}  \quad & \stermt{rs} & \defeq & \stermm{\varnothing} \mid \stermm{\{r_1,\dotsc,r_n\}} \mid \stermm{\rs_1\uplus\rs_2} \mid \stermm{\rs_1\cup\rs_2} \mid \stermm{\rs_1\cap\rs_2} \mid \cdots \\
\text{session types}     \quad & \stypem(\pi)                & \defeq & \stypem({\nil(r)}) \mid \stypem({\msg(r,\tau)\seqs\pi})  \mid \stypem({\msg(r_1,r_2,\hat\tau)\seqs\pi}) \mid \stypem({\quan(r,\lambda a{:}\sigma.\pi)})\\

\text{linear base types} \quad & \typem({\hat\delta(\vv s)}) & \defeq & \cdots             \mid \chan{\rs}{\pi} \\
\text{dynamic constant resources} \quad & \dcr & \defeq & \cdots \mid \dtermm{c^{\rs}} \\
\text{signature} \quad & \sig & \defeq & \cdots \mid \sig, \dtermm{c} : \stypem(\pi)

\end{array}
\]

Additional Signature of Static Constants
\[
\begin{array}{rclrcl}
\typet(chan)         & : & \sortm({(\stset,\ststype)\Rightarrow\stvtype})             \\
\nil                 & : & \sortm({(\stint) \Rightarrow \ststype})                           &
\msg                 & : & \sortm({(\stint,\sttype,\ststype)\Rightarrow\ststype})           \\
\msg                 & : & \sortm({(\stint,\stint,\stvtype,\ststype)\Rightarrow\ststype})    &
\quan                & : & \sortm({(\stint,\stsigma\rightarrow\ststype)\Rightarrow\ststype}) 
\end{array}
\]

Additional Typings
\begin{gather*}
\prftree[r]{\rulename{sig-chan}}
    {\sig,\dtermm{c}:\stypem(\pi)\vDash\dtermm{c^{\rs}}:\chan{\rs}{\pi}}
\end{gather*}

Additional Signature of Dynamic Constants\\
\end{adjustbox}

\begin{adjustbox}{varwidth=1.5\textwidth,max width=\textwidth}
\[
\begin{array}{rcl}
\dfork   &:& \typem({\forall\rs_1,\rs_2{:}\stsetc.
                     \forall\pi{:}\ststypec.
                     (\rs_1\uplus\rs_2=\fullset)\sguard}) 
             \typem({(\chan{\rs_1}{\pi}\multimap\sunit)
                     \Rightarrow\chan{\rs_2}{\pi}}) \\
\dcut    &:& \typem({\forall\rs_1,\rs_2{:}\stsetc.
                     \forall\pi{:}\ststypec.
                     (\rs_1\cup\rs_2=\fullset)\sguard}) 
             \typem({(\chan{\rs_1}{\pi}),\chan{\rs_2}{\pi})
                     \Rightarrow\chan{\rs_1\cap\rs_2}{\pi}}) \\
\delim   &:& \typem({
                     \forall\pi{:}\ststypec.})
             \typem({\chan{\varnothing}{\pi}
                     \Rightarrow\sunit}) \\   
\dsplit  &:& \typem({\forall\rs_1,\rs_2{:}\stsetc.
                     \forall\pi{:}\ststypec.
                     (\rs_1\cap\rs_2=\varnothing)\sguard}) 
             \typem({(\chan{\rs_1\uplus\rs_2}{\pi},\chan{\rs_1}{\pi}\multimap\sunit)
                     \Rightarrow\chan{\rs_2}{\pi}}) \\
\dbsend  &:& \typem({\forall\rs{:}\stsetc.
                     \forall r{:}\stintc.
                     \forall \pi{:}\ststypec.
                     \forall \tau{:}\sttypec.
                     (r\in\rs)\sguard}) 
             \typem({(\chan{\rs}{\msg(r,\tau)\seqs\pi},\tau)
                     \Rightarrow\chan{\rs}{\pi}}) \\
\dbrecv  &:& \typem({\forall\rs{:}\stsetc.
                     \forall r{:}\stintc.
                     \forall\pi{:}\ststypec.
                     \forall\tau{:}\sttypec.
                     (r\notin\rs)\sguard}) 
             \typem({\chan{\rs}{\msg(r,\tau)\seqs\pi}
                     \Rightarrow\chan{\rs}{\pi}\otimes\tau}) \\
\dsend   &:& \typem({\forall\rs{:}\stsetc.
                     \forall r_1,r_2{:}\stintc.
                     \forall \pi{:}\ststypec.
                     \forall \hat\tau{:}\stvtypec. }) \typem({
                     (r_1\in\rs\land r_2\notin\rs)\sguard}) 
             \typem({(\chan{\rs}{\msg(r_1,r_2,\hat\tau)\seqs\pi},\hat\tau)
                     \Rightarrow\chan{\rs}{\pi}}) \\
\drecv   &:& \typem({\forall\rs{:}\stsetc.
                     \forall r_1,r_2{:}\stintc.
                     \forall \pi{:}\ststypec.
                     \forall \hat\tau{:}\stvtypec. })  \typem({
                     (r_1\notin\rs\land r_2\in\rs)\sguard}) 
             \typem({\chan{\rs}{\msg(r_1,r_2,\hat\tau)\seqs\pi}
                     \Rightarrow\chan{\rs}{\pi}\tensor\hat\tau}) \\
\dskip   &:& \typem({\forall\rs{:}\stsetc.
                     \forall r_1,r_2{:}\stintc.
                     \forall \pi{:}\ststypec.
                     \forall \hat\tau{:}\stvtypec. })  \typem({
                     ((r_1,r_2\in\rs)\lor (r_1,r_2\notin\rs))\sguard}) 
             \typem({\chan{\rs}{\msg(r_1,r_2,\hat\tau)\seqs\pi}
                     \Rightarrow\chan{\rs}{\pi}\tensor\hat\tau}) \\
\dclose  &:& \typem({\forall\rs{:}\stsetc.\forall r{:}\stintc.(r\in\rs)\sguard\chan{\rs}{\nil(r)}\Rightarrow\sunit}) \\
\dwait   &:& \typem({\forall\rs{:}\stsetc.\forall r{:}\stintc.(r\notin\rs)\sguard\chan{\rs}{\nil(r)}\Rightarrow\sunit}) \\
\dunify  &:& \typem({\forall\rs{:}\stsetc.
                     \forall r{:}\stintc.
                     \forall f{:}\sigma\rightarrow\ststypec.
                     (r\in\rs)\sguard}) 
             \typem({\chan{\rs}{\quan(r,f)}
                     \Rightarrow\forall a{:}\sigma.\chan{\rs}{f(a)}}) \\
\dexify  &:& \typem({\forall\rs{:}\stsetc.
                     \forall r{:}\stintc.
                     \forall f{:}\sigma\rightarrow\ststypec.
                     (r\notin\rs)\sguard}) 
             \typem({\chan{\rs}{\quan(r,f)}
                     \Rightarrow\exists a{:}\sigma.\chan{\rs}{f(a)}}) \\
\end{array}
\]
\end{adjustbox}
\caption{Additional Syntax and Typings of $\Ldeppi$}
\label{fig:mpdepsyntax}
\end{figure}

%% file: mpdepred.tex
\begin{figure}
\begin{adjustbox}{varwidth=1.5\textwidth,max width=\textwidth}
Redexes
\[
\begin{array}{rcl}
\text{partial (ad-hoc) redex} & \defeq & \dtermm{\dsend(c^{\rs},v)}  \mid \dtermm{\drecv(c^{\rs})} \mid \dtermm{\dskip(c^{\rs})}  \mid \dtermm{\dbsend(c^{\rs},v)} \mid \dtermm{\dbrecv(c^{\rs})} \mid \dtermm{\dclose(c^{\rs})} \mid \dtermm{\dwait(c^{\rs})} \mid  \dtermm{\dunify(c^{\rs})}      \mid \dtermm{\dexify(c^{\rs})}
\end{array}
\]

Pool Reductions
\[
\begin{array}{lrcl}
\text{(\rulename{pr-fork})}  & \dtermm{\dpool,t{:}E[\dfork(\dlam x.e)]}                                                                        & \creduce & \dtermm{\dpool,t{:}E[c^{\rs_2}], t'{:}e\map(x,c^{\rs_1})}                                               \\
\text{(\rulename{pr-cut})}   & \dtermm{\dpool,t{:}E[\dcut(c_1^{\rs_1}, c_2^{\rs_2})]}                                                          & \creduce & \dtermm{\dpool\map({c_1,c_2},{c,c}), t{:}E[c^{\rs_1\cap\rs_2}]}                                         \\
\text{(\rulename{pr-elim})}  & \dtermm{\dpool,t{:}E[\delim(c^{\varnothing})]}                                                                          & \creduce & \dtermm{\dpool,t{:}E[\dunit]}                                                                           \\
\text{(\rulename{pr-split})} & \dtermm{\dpool,t{:}E[\dsplit(c^{\rs_1\uplus\rs_2},\dlam x.e)]}                                                  & \creduce & \dtermm{\dpool,t{:}E[c^{\rs_2}], t'{:}e\map(x,c^{\rs_1} )}                                              \\
\text{(\rulename{pr-bmsg})}  & \dtermm{\dpool,t_1{:}E[\dbsend(c^{\rs_1}, v)], t_2{:}E[\dbrecv(c^{\rs_2})],\dotsc, t_n{:}E[\dbrecv(c^{\rs_n})]} & \creduce & \dtermm{\dpool,t_1{:}E[c^{\rs_1}], t_2{:}E[\dpair(c^{\rs_2}, v)],\dotsc, t_n{:}E[\dpair(c^{\rs_n}, v)]} \\
\text{(\rulename{pr-end})}   & \dtermm{\dpool,t_1{:}E[\dclose(c^{\rs_1})], t_2{:}E[\dwait(c^{\rs_2})],\dotsc, t_n{:}E[\dwait(c^{\rs_n})]}      & \creduce & \dtermm{\dpool,t_1{:}E[\dunit],t_2{:}E[\dunit],\dotsc, t_n{:}E[\dunit]}         \\                        
\text{(\rulename{pr-quan})}  & \dtermm{\dpool,t_1{:}E[\dunify(c^{\rs_1})], t_2{:}E[\dexify(c^{\rs_2})],\dotsc, t_n{:}E[\dexify(c^{\rs_n})]}    & \creduce & \dtermm{\dpool,t_1{:}E[c^{\rs_1}],t_2{:}E[c^{\rs_2}],\dotsc, t_n{:}E[c^{\rs_n}]} \\[1.5\jot]
\text{(\rulename{pr-msg})}   & \stackbox[r][c]{$\dtermm{\dpool,t_1{:}E[\dsend(c^{\rs_1}, v)], t_2{:}E[\drecv(c^{\rs_2})],}$\\ $\dtermm{t_3{:}E[\dskip(c^{\rs_3})], \dotsc, t_n{:}E[\dskip(c^{\rs_n})]}$} & \creduce &
                               \stackbox[l][c]{$\dtermm{\dpool,t_1{:}E[c^{\rs_1}],t_2{:}E[\dpair(c^{\rs_2}, v)],}$\\ $\dtermm{t_3{:}E[c^{\rs_3}],\dotsc,t_n{:}E[c^{\rs_n}]}$ } 
\end{array}
\]


\begin{center}
\small{For the last four pool reduction rules, we assume $\endpoints(\dtermm{c})=\{\dtermm{c^{\rs_1}},\dotsc,\dtermm{c^{\rs_n}}\}$.}
\end{center}

Pool Equivalences
\[
\begin{array}{lrcll}
\text{(\rulename{pe-cut})}   & \dtermm{\dpool,t{:}E[\dcut(x, y)]}    & \equiv & \dtermm{\dpool,t{:}E[\dcut(y, x)]}    
\end{array}
\]
\end{adjustbox}
\caption{Additional Reductions in $\Ldeppi$}
\label{fig:mpdepred}
\end{figure}

%% file: bundle.bbl
\begin{thebibliography}{10}
\providecommand{\url}[1]{\texttt{#1}}
\providecommand{\urlprefix}{URL }
\providecommand{\doi}[1]{https://doi.org/#1}

\bibitem{Abramsky:1993io}
Abramsky, S.: {Computational Interpretations of Linear Logic.} Theor. Comput.
  Sci.  (1993)

\bibitem{Abramsky:1994ez}
Abramsky, S.: {Proofs as Processes.} Theor. Comput. Sci.  (1994)

\bibitem{Bejleri:2009eo}
Bejleri, A., Yoshida, N.: {Synchronous Multiparty Session Types.} Electr. Notes
  Theor. Comput. Sci.  (2009)

\bibitem{Bellin:1994ua}
Bellin, G., Scott, P.J.: {On the pi-Calculus and Linear Logic.} Theor. Comput.
  Sci.  (1994)

\bibitem{Caires:2010gi}
Caires, L., Pfenning, F.: {Session Types as Intuitionistic Linear
  Propositions.} In: CONCUR (2010)

\bibitem{Carbone:2016kd}
Carbone, M., Lindley, S., Montesi, F., Sch{\"u}rmann, C., Wadler, P.:
  {Coherence Generalises Duality - A Logical Explanation of Multiparty Session
  Types.} CONCUR  (2016)

\bibitem{Carbone:2015hl}
Carbone, M., Montesi, F., Sch{\"u}rmann, C., Yoshida, N.: {Multiparty Session
  Types as Coherence Proofs.} CONCUR  (2015)

\bibitem{Carbone:2017ki}
Carbone, M., Montesi, F., Sch{\"u}rmann, C., Yoshida, N.: {Multiparty session
  types as coherence proofs.} Acta Inf.  (2017)

\bibitem{Gay:2005gc}
Gay, S.J., Hole, M.: {Subtyping for session types in the pi calculus.} Acta
  Inf.  (2005)

\bibitem{Honda:1993eh}
Honda, K.: {Types for Dyadic Interaction.} In: CONCUR (1993)

\bibitem{Honda:1998fm}
Honda, K., Vasconcelos, V.T., Kubo, M.: {Language Primitives and Type
  Discipline for Structured Communication-Based Programming.} In: ESOP (1998)

\bibitem{Honda:2008hi}
Honda, K., Yoshida, N., Carbone, M.: {Multiparty asynchronous session types.}
  POPL  (2008)

\bibitem{Lindley:2015cr}
Lindley, S., Morris, J.G.: {A Semantics for Propositions as Sessions.} ESOP
  \textbf{9032}(Chapter 23),  560--584 (2015)

\bibitem{Pfenning:2011ce}
Pfenning, F., Caires, L., Toninho, B.: {Proof-Carrying Code in a Session-Typed
  Process Calculus.} CPP  (2011)

\bibitem{Takeuchi:1994bv}
Takeuchi, K., Honda, K., Kubo, M.: {An Interaction-based Language and its
  Typing System.} In: PARLE (1994)

\bibitem{Toninho:2011in}
Toninho, B., Caires, L., Pfenning, F.: {Dependent session types via
  intuitionistic linear type theory.} PPDP  (2011)

\bibitem{Wadler:2012ua}
Wadler, P.: {Propositions as sessions.} ICFP  (2012)

\bibitem{Wu:2017wd}
Wu, H., Xi, H.: {Dependent Session Types.} CoRR  (2017)

\bibitem{Xi:1998wa}
Xi, H.: {Dependent Types in Practical Programming}. Ph.D. thesis, Carnegie
  Mellon University, Pittsburgh, PA (1998)

\bibitem{Xi:2003kl}
Xi, H.: {Applied Type System - Extended Abstract.} TYPES  (2003)

\bibitem{Xi:2007te}
Xi, H.: {Dependent ML An approach to practical programming with dependent
  types.} J. Funct. Program.  (2007)

\bibitem{Xi:2016vd}
Xi, H.: {\emph{Applied Type System: An Approach to Practical Programming with
  Theorem-Proving}}. ats-lang.org  (Jan 2016)

\bibitem{Xi:1999bh}
Xi, H., Pfenning, F.: {Dependent Types in Practical Programming.} POPL  (1999)

\bibitem{Xi:2016ue}
Xi, H., Ren, Z., Wu, H., Blair, W.: {Session Types in a Linearly Typed
  Multi-Threaded Lambda-Calculus.} CoRR  (2016)

\bibitem{Xi:2017wv}
Xi, H., Wu, H.: {Multirole Logic (Extended Abstract).} CoRR  (2017)

\bibitem{Yoshida:2010ht}
Yoshida, N., Deni{\'e}lou, P.m., Bejleri, A., Hu, R.: {Parameterised Multiparty
  Session Types.} FoSSaCS  \textbf{6014}(Chapter 10),  128--145 (2010)

\bibitem{Yoshida:2007hk}
Yoshida, N., Vasconcelos, V.T.: {Language Primitives and Type Discipline for
  Structured Communication-Based Programming Revisited - Two Systems for
  Higher-Order Session Communication.} Electr. Notes Theor. Comput. Sci.
  (2007)

\end{thebibliography}
